\documentclass[preprint,1p]{elsarticle}

\usepackage{framed,multirow}

\journal{Applied Numerical Mathematics}

\usepackage{amsthm}
\usepackage{amsmath}
\usepackage{amsfonts}
\usepackage{amssymb}
\usepackage{mathtools}
\usepackage{subfigure}
\usepackage{xcolor}

\graphicspath{{figures/}}

\theoremstyle{plain}

\newtheorem{lemma}{Lemma}
\newtheorem{theorem}{Theorem}
\newtheorem{corollary}{Corollary}
\newtheorem{proposition}{Proposition}

\theoremstyle{definition}

\theoremstyle{remark}
\newtheorem{remark}{Remark}

\newcommand{\bome}{\boldsymbol\omega} 

\newcommand{\R}{\mathbb R}

\newcommand{\sinc}{\operatorname{sinc}}
\newcommand{\cfun}{\text{circ}}
\newcommand{\rect}{\operatorname{rect}}

\newcommand{\intinf}{\int_{-\infty}^{\infty}}
\newcommand{\dblintinf}{\int_{-\infty}^{\infty}\int_{-\infty}^{\infty}}
\newcommand{\intW}{\int_{-W}^{W}}
\newcommand{\dblintW}{\int_{-W}^{W}\int_{-W}^{W}}

\newcommand{\del}{\partial}
\newcommand{\ome}{\omega}

\newcommand{\twoFigWidth}{0.4}

\begin{document}


\begin{frontmatter}

\title{Diffraction integral computation using sinc approximation}

\author[1]{Max Cubillos\corref{CorrAuthor}} \ead{afrl.rdl.orgbox@us.af.mil}
\author[2]{Edwin Jimenez} \ead{jimenez@caltech.edu}

\cortext[CorrAuthor]{Corresponding author}

\address[1]{Directed Energy Directorate, Air Force Research Laboratory, 3550 Aberdeen Ave SE, Kirtland AFB, Albuquerque, New Mexico 87117, USA}
\address[2]{Department of Computing and Mathematical Sciences, California Institute of Technology, 1200 E California Blvd, MC 305-16, Pasadena, California 91125, USA}

\begin{abstract}
We propose a method based on sinc series approximations for computing the Rayleigh-Sommerfeld and Fresnel diffraction integrals of optics. The diffraction integrals are given in terms of a convolution, and our proposed numerical approach is not only super-algebraically convergent, but it also
satisfies an important property of the convolution---namely, the preservation
of bandwidth.  Furthermore, the accuracy of the proposed method depends only on
how well the source field is approximated; \emph{it is independent of
wavelength, propagation distance, and observation plane discretization}. In
contrast, methods based on the fast Fourier transform (FFT), such as the
angular spectrum method (ASM) and its variants, approximate the optical fields
in the source and observation planes using Fourier series. We will show that
the ASM introduces artificial periodic boundary conditions and violates the
preservation of bandwidth property, resulting in limited accuracy
which decreases for longer propagation distances.  The sinc-based approach
avoids both of these problems.  Numerical results are presented
for Gaussian beam propagation and circular aperture diffraction to demonstrate
the high-order accuracy of the sinc method for both short-range and long-range
propagation. For comparison, we also present numerical results obtained with
the angular spectrum method.  
\end{abstract}


\end{frontmatter}

\section{Introduction} \label{sec:Introduction}

One of the most basic problems in optics is the propagation of a known scalar
optical field from one plane to another. Solutions to this problem are given in
terms of  various integrals---such as the Kirchoff, Rayleigh-Sommerfeld,
Fresnel, and Fraunhofer integrals---with differing regimes of validity. These
integrals have only a few analytical solutions, and therefore most optical
propagation problems rely on numerical methods.

A variety of methods have been developed to evaluate the integrals of optics
but the most popular are based on the convolution theorem and the FFT. In this
article, we will compare our proposed approach to the standard FFT-based
method: the angular spectrum method
(ASM)~\cite{Sherman1967Int,ShewellWolf1968,Goodman1996,schmidt_numerical_2010}.
The ASM is also known by other names in the literature (such as the transfer
function propagator~\cite{voelz_computational_2011}) and the name ``angular
spectrum method'' is sometimes also used to refer to different methods.  In
optical propagation the unknown field at some observation plane is written in
terms of a convolution of a diffraction kernel and a source field, and the ASM
simply performs this computation as the inverse discrete Fourier transform of
the product of the \emph{known} Fourier transform of the kernel and the
numerically computed Fourier transform of the source
field~\cite{voelz_computational_2011,schmidt_numerical_2010}.

The fundamental problem with the ASM---as with other FFT-based methods---is
that it approximates the optical field by a Fourier series, which introduces
artificial periodicity into the domain.  This artificial periodicity leads to
errors from modes in the solution that should be scattered to infinity but are
instead periodically reintroduced into the computational domain.  The most
straightforward way to reduce the errors caused by artificial periodicity is to
increase the size of the domain, but at the cost of additional computation.
There have been attempts to repair FFT-based methods, such as the band-limited
angular spectrum method~\cite{matsushima_band-limited_2009}; see also the
recent article~\cite{zhang_analysis_2020} and the references therein.
Unfortunately, all such modifications only mitigate the artefacts of artificial
periodicity under certain circumstances; they cannot eliminate the underlying
problem entirely.

The problems that arise from using a Fourier basis for an infinite domain are
naturally remedied by using a suitable set of basis functions such as the
Whittaker cardinal basis or sinc functions, which are analytic in the whole
complex plane.  Methods based on sinc functions have a long history, dating
back more than 100 years to the work of Borel~\cite{borel_sur_1897} and
Whittaker~\cite{whittaker_functions_1915}. Since then, sinc methods have been
used in a wide range of applications, including approximation theory, image
processing, and the numerical solution of integral and differential equations;
see the book~\cite{Stenger2012} and references therein for a review of
sinc-based methods. Although they have been used in a variety of computational
physics problems, to our knowledge, sinc methods have not been applied to the
problem of optical propagation. 

In this article, we propose a new method for computing diffraction integrals
based on approximation by sinc functions. For optical fields that are
approximately bandlimited and decay exponentially in the transverse spatial
direction, approximation by sinc functions is super-algebraically
convergent~\cite[Chapter 1]{Stenger2012}.  Furthermore, the only source of
error is the approximation of the optical field; the error for computing the
diffraction integral does not depend on wavelength, propagation distance, or
observation plane grid.  Although the method allows for arbitrary sample
spacing in the source and observation planes, when the sampling is the same in
both planes (a requirement of ASM), the FFT can be used to accelerate the
computation. Thus, in this case, our sinc method has the same computational
complexity as the ASM.  However, in the near field, for an $N$-point source
grid, our sinc-based algorithm can be reduced to $\mathcal O(N)$ complexity,
making it faster than any FFT-based method.

This article is organized as follows: Section~\ref{sec:Diffraction} reviews the
diffraction integrals covered in this article, emphasizing some of the
important properties and their relevance to numerical methods.
Section~\ref{sec:Pitfalls} discusses the ASM in the context of the arguments
made in the preceding section together with related partial differential
equations, showing why it is an inadequate method. Section~\ref{sec:SincSeries}
describes in detail the sinc-based method proposed in this article and shows
how it naturally prevails where the ASM falls short. Finally,
Section~\ref{sec:NumericalResults} presents numerical results showing, in
practice, the advantages of the sinc method over the ASM.

\section{Diffraction integrals and the convolution theorem} \label{sec:Diffraction}

We are interested in the propagation of a scalar optical field $u(x,y)$ from a
given source plane to an observation plane a distance $z$ away. We will use
lower case letters $(u, x, y)$ to denote fields and coordinates in the source
plane; upper case letters $(U, X, Y)$ will denote those quantities in the
observation plane. For brevity, we will denote the observation field $U(X,Y,z)$
by $U_z(X,Y)$ or simply $U(X,Y)$ when the distance from the source to the
observation plane is clear from the context. The solution to the propagation
problem is given by the Rayleigh-Sommerfeld diffraction integral
\begin{equation} \label{eq:RayleighSommerfeld}
  U(X,Y) = \frac{-1}{2\pi} \dblintinf \frac{\partial}{\partial z}\left(\frac{\exp \left(ik\sqrt{(X-x)^2 + (Y-y)^2 + z^2}\right)}{\sqrt{(X-x)^2 + (Y-y)^2 + z^2}}\right) u(x,y)\,dx\,dy,
\end{equation}
where $k = 2\pi / \lambda$ is the wavenumber and $\lambda$ is the wavelength.
Using the paraxial approximation (see, e.g.,~\cite{schmidt_numerical_2010}),
the Rayleigh-Sommerfeld integral can be approximated by the Fresnel diffraction
integral
\begin{equation} \label{eq:FresnelIntegral}
  U(X,Y) = \frac{-ik e^{ikz}}{2\pi z} \dblintinf e^{i\frac k{2z}((X-x)^2 + (Y-y)^2)} u(x,y)\,dx\,dy.
\end{equation}
Denoting the convolution of two functions $f$ and $g$ by
\begin{equation}
  (f \ast g)(X,Y) = \dblintinf f(X-x,Y-y) g(x,y)\,dx\,dy,
\end{equation}
we can write the Rayleigh-Sommerfeld and Fresnel diffraction integrals in the
compact form
\begin{equation}\label{eq:DiffractionConvolution}
  U(X,Y) = (h \ast u)(X,Y), 
\end{equation}
where $h$ is the diffraction kernel. The Rayleigh-Sommerfeld kernel is
given by
\begin{equation}
  h(x,y) = h_{RS}(x,y) = \frac{-1}{2\pi} \frac{\partial}{\partial z}\left(\frac{\exp \left(ik\sqrt{x^2 + y^2 + z^2}\right)}{\sqrt{x^2 + y^2 + z^2}}\right),
\end{equation}
and the Fresnel kernel is 
\begin{equation}
  h(x,y) = h_F(x,y) = \frac{-ik e^{ikz}}{2\pi z} \exp\left(\frac{ik}{2z}(x^2 + y^2) \right).
\end{equation}
Much of what follows will apply to both kernels, in which case we will use $h$
to denote either one; where specificity is required we will use the subscript
$RS$ or $F$.

The Fourier transform will play a crucial role in what follows. We will use the
operator notation $\mathcal F[f](\xi,\eta)$ (and $\mathcal F^{-1}$ for the
inverse) and $\widehat f(\xi,\eta)$ interchangeably for the Fourier transform
of a two-dimensional function $f(x,y)$:
\begin{align}
  \widehat f(\xi,\eta) = \mathcal F[f](\xi,\eta) 
    &= \dblintinf f(x,y) e^{-i 2\pi (x\xi+y\eta)} dx\,dy, \\
  f(x,y) = \mathcal F^{-1}[\widehat f](x,y) 
    &= \dblintinf \widehat f(\xi,\eta) e^{i 2\pi (x\xi+y\eta)} d\xi\,d\eta.
\end{align}
We will make frequent use of the well-known convolution theorem for Fourier transforms:
\begin{theorem}
Let $f$ and $g$ be functions with Fourier transforms $\widehat f$ and $\widehat
g$, respectively. Then
\begin{subequations}\label{eq:ConvThm}
\begin{align}
  \mathcal F[f \ast g](\xi,\eta)  &= \widehat f(\xi,\eta) \, \widehat g(\xi,\eta), 
    \label{eq:FofConv} \\
  \mathcal F[f \, g](\xi,\eta) &= \left( \widehat f \ast \widehat g \right) (\xi,\eta)
    \label{eq:FofProd}.
\end{align}
\end{subequations}
\end{theorem}
\begin{proof}
See~\cite[Chapter 5]{DebnathMikusinskiBook}.
\end{proof}
It can be shown that the Fourier transform of the Rayleigh-Sommerfeld kernel is
given by~\cite{lalor_conditions_1968} 
\begin{equation} \label{eq:RSFT}
  \widehat h_{RS}(\xi,\eta) = \exp\left(i [k^2 - 4\pi^2(\xi^2+\eta^2)]^{1/2} z\right),
\end{equation}
where a branch of the complex square root $z^{1/2}$ is chosen such that 
\begin{align}
  z^{1/2} = 
    \begin{cases} 
      \sqrt{x}, & z = x \in (0,\infty) \\ 
      i\sqrt{|x|}, & z = x \in (-\infty,0)
    \end{cases}.
\end{align}
The Fourier transform of the Fresnel kernel, on the other hand, is
\begin{equation} \label{eq:FresnelFT}
  \widehat h_F(\xi,\eta) = e^{ikz} \exp\left(-i\frac{2\pi^2z}{k}(\xi^2 + \eta^2)\right).
\end{equation}
\begin{remark}\label{rem:hAmp}
Note that $\widehat h_F$ has unit amplitude for all $\xi$ and $\eta$, whereas
$\widehat h_{RS}$ has unit amplitude only for $4\pi^2(\xi^2+\eta^2) \leq k^2$.
\end{remark}

The angular spectrum method is based on an application of~\eqref{eq:FofConv} to 
equation~\eqref{eq:DiffractionConvolution} so that
\begin{equation}
  U(X,Y) = \mathcal F^{-1} \left[ \widehat h \, \widehat u \right](X,Y).
\end{equation}
Thus, to compute either~\eqref{eq:RayleighSommerfeld}
or~\eqref{eq:FresnelIntegral}, the ASM requires that we compute $\widehat{u}$
numerically, multiply it by~\eqref{eq:RSFT} or~\eqref{eq:FresnelFT}, and then
apply the inverse discrete Fourier transform to the result.

\medskip

We note that the convolution theorem has two particularly important
consequences for \emph{bandlimited functions}, which we define as follows. A
function $f$ has bandwidth $W$ if $(-W,W)^2$, $W > 0$, is the smallest square that contains the support of $\widehat{f}$, and it is bandlimited if $W < \infty$.  Similarly, when the function itself (as opposed to its Fourier transform) vanishes outside of a square $(-D,D)^2$, $D > 0$, we say that it has support width $D$. For convenience, we record two simple consequences of~\eqref{eq:ConvThm} for the convolution of two functions.
\begin{corollary} \label{cor:Bandwidth}
Let $f$ and $g$ be functions with bandwidths $W_f$ and $W_g$ and support widths
$D_f$ and $D_g$, respectively, which may be infinite. Then
\begin{enumerate}
  \item $f \ast g$ has bandwidth $\min(W_f,W_g)$.
  \item $f \ast g$ has support width $D_f + D_g$.
\end{enumerate}
\end{corollary}
\begin{proof}
This follows easily from our definitions of bandwidth and support width.
\end{proof}
Since the bandwidth of a convolution is equal to the smaller bandwidth of the
two functions being convolved, if Corollary~\ref{cor:Bandwidth} is applied to
either the Rayleigh-Sommerfeld or Fresnel diffraction integral, it follows that
\emph{optical propagation preserves bandwidth}: the observation plane field $U$
will have the same bandwidth as the source plane field $u$. Furthermore, for
Fresnel diffraction, because $\widehat h_F$ has unit amplitude, not only is the
bandwidth preserved but the amplitude of $\widehat U$ at each frequency is the
same as $\widehat u$. In other words, Fresnel propagation preserves the
frequency-space envelope.

Conversely, the corollary also implies that \emph{optical propagation does not
preserve support width}; even if $u$ has bounded support, convolution with
$h_{RS}$ or $h_F$ (both of which have unbounded support) will cause $U$ to have
unbounded support. Ideally, a numerical method for computing the diffraction
integrals should respect these properties of the convolution.

\section{Pitfalls of Fourier series approximations of diffraction integrals} \label{sec:Pitfalls}

As previously mentioned, the diffraction integrals have the effect of taking an
optical field $u$ with bounded support and transforming it into a function
$U$ with unbounded support. To understand the nature of this transformation,
we consider certain properties of the related partial differential equations
and how they relate to their respective integral solutions.

\subsection{The Fresnel case} \label{sec:FresnelFourier}
Suppose that the field $U$ propagates along the $+z$-direction so that it takes
the form
\begin{equation}\label{eq:Ueikz}
  U(x,y,z) = e^{ikz} \psi(x,y,z),
\end{equation}
where $\psi$ is a slowly-varying envelope function.  The paraxial/parabolic
wave equation (see, e.g.,~\cite[Chapter 11]{Blackledge2005digital}) for $\psi$
is obtained under the assumption that $U$ satisfies the Helmholtz equation
\begin{equation}\label{eq:Helmholtz}
  \Delta U + k^2 U = 0,
\end{equation}
where $\Delta = \del_x^2 + \del_y^2 + \del_z^2$ denotes the Laplacian, and
$\psi$ varies slowly with respect to $z$ so that
\begin{equation}\label{eq:SlowPsiAssump}
  \left| \frac{\del^2 \psi}{\del z^2} \right| \ll
  2 k \left| \frac{\del \psi}{\del z} \right|.
\end{equation}
Then, the Helmholtz equation and assumptions~\eqref{eq:Ueikz}
and~\eqref{eq:SlowPsiAssump} lead to the paraxial equation
\begin{subequations}\label{eq:Paraxial}
\begin{alignat}{2} 
  -2ik \frac{\del \psi}{\del z} &= \Delta_{\perp} \psi, 
   &\qquad &(x,y,z) \in \R^2 \times (0,\infty), \label{eq:ParaxialEqn} \\
   \psi(x,y,0) &= u(x,y), &&(x,y) \in \R^2, \label{eq:ParaxialIC} 
\end{alignat} 
\end{subequations} 
where $\Delta_{\perp} = \del_x^2 + \del_y^2$ is the transverse Laplacian and
the source field $u(x,y)$ is taken as an initial condition. Note that the
Fresnel integral, and the assumed form for $U$ given in~\eqref{eq:Ueikz},
yields a solution $\psi$ to~\eqref{eq:Paraxial}.

\medskip

Letting $\psi$ equal a single Fourier mode, $\psi = \exp(2\pi i(\xi x + \eta y
- \omega z))$, and substituting in~\eqref{eq:ParaxialEqn}, we obtain the
dispersion relation
\begin{equation}
  \omega = \frac{\pi}k (\xi^2 + \eta^2).
\end{equation}
Now, taking the magnitude of the gradient of $\ome$ with respect to $\xi$ and
$\eta$, we get the group velocity
\begin{equation}
  v_g = |\nabla_{\xi,\eta}\omega| = \frac{2\pi}k \sqrt{\xi^2 + \eta^2},
\end{equation}
which is the speed at which a wave packet will travel in the transverse
direction.

We observe that the group velocity is directly proportional to spatial
frequency. This means that the higher the spatial frequencies present, the
further they will travel in the transverse direction.  That is why, no matter
how tightly focused an initial beam, for example, there will always be some
sort of distortion and spreading.  Additionally, we note that the dispersion
relation has no imaginary part for all frequencies---i.e., waves travel off to
infinity in the transverse direction with no dissipation. This is another
manifestation of the fact we noted earlier that Fresnel propagation preserves
the frequency-space envelope.  If periodic boundary conditions are imposed
together with equation~\eqref{eq:Paraxial} , however, the transverse waves
obviously do not disperse to infinity; they keep circling around in the domain
forever.

\medskip

Next, we will demonstrate that solving~\eqref{eq:Paraxial} with periodic
boundary conditions is equivalent to approximating the Fresnel diffraction
integral with the ASM.

For some real $D > 0$, let the domain be $[-D/2,D/2]^2$ and suppose that the
initial condition is given by the finite Fourier series
\begin{equation} \label{eq:InitialCondition}
  u(x,y) = \sum_{m=-N/2}^{N/2} \sum_{n=-N/2}^{N/2} \widehat v_{mn} 
           e^{i2\pi (\xi_m x + \eta_n y)},
\end{equation}
where $\xi_m = m/D$, $\eta_n = n/D$, for $m,n=-N/2,\dotsc,N/2$, and $N$ is an
even positive integer. Then, we can also write the solution $\psi$
to~\eqref{eq:Paraxial} as a Fourier series
\begin{equation} \label{eq:PsiSeries}
  \psi(x,y,z) = \sum_{m=-N/2}^{N/2} \sum_{n=-N/2}^{N/2} \widehat \psi_{mn}(z) 
                  e^{i2\pi (\xi_m x + \eta_n y)}.
\end{equation}
Substituting the expression~\eqref{eq:PsiSeries} into the paraxial equation, we
find that each Fourier coefficient in the series satisfies
\begin{equation}
  -2ik \frac{d \widehat \psi_{mn}}{d z} = (2\pi)^2 (\xi_m^2 + \eta_n^2) \widehat \psi_{mn}.
\end{equation}
Integrating from $0$ to $z$, we get
\begin{align}
  \widehat \psi_{mn}(z) = \exp\left( -i \frac{2\pi^2}k (\xi_m^2 + \eta_n^2) z \right) 
                          \widehat v_{mn},
\end{align}
and the solution $U(x,y,z)$ can be recovered by multiplying the
series~\eqref{eq:PsiSeries} by $e^{ikz}$ which, using the notation discussed in
Section~\ref{sec:Diffraction}, becomes
\begin{equation} \label{eq:ParaxialSolution}
  U(X,Y) = e^{ikz} \sum_{m=-N/2}^{N/2} \sum_{n=-N/2}^{N/2} \widehat v_{mn} \exp\left( -i \frac{2\pi^2}k (\xi_m^2 + \eta_n^2) z \right) e^{ i2\pi (\xi_m x + \eta_n y)}.
\end{equation}

On the other hand, invoking the convolution theorem, the ASM computes the
Fresnel integral~\eqref{eq:FresnelIntegral} as
\begin{equation}\label{eq:ASMInvFourier}
\begin{split}
  U(X,Y) &= \mathcal F^{-1} \left[ \widehat h_F \, \widehat u \right](X,Y) \\
         &= e^{ikz} \dblintinf e^{i2\pi (\xi X + \eta Y)} 
                           \exp\left( -i\frac{2\pi^2z}k (\xi^2 + \eta^2) \right) 
                            \widehat u(\xi,\eta)\,d\xi\,d\eta. 
\end{split}
\end{equation}
The first step of the ASM is to represent the initial profile $u$ by a Fourier
series over $[-D/2,D/2]^2$ and by $0$ elsewhere.  The Fourier transform
of~\eqref{eq:InitialCondition} is given by
\begin{equation} \label{eq:SincSeries}
  \widehat u(\xi,\eta) = \sum_{m=-N/2}^{N/2} \sum_{n=-N/2}^{N/2} D^2 \, \widehat v_{mn}\,
                         \sinc\left[ D(\xi - \xi_m)   \right] 
                         \sinc\left[ D(\eta - \eta_n) \right],
\end{equation}
where we use the normalized sinc function
\begin{equation}\label{eq:NormalizedSinc}
  \sinc(x) = \frac{\sin(\pi x)}{\pi x},
\end{equation}
and we have used the fact that the transform of a Fourier component is a scaled
and shifted sinc function. It follows that the Fourier transform is
unbounded---i.e., \emph{a function with bounded support has infinite
bandwidth}. However, we note that 
\begin{equation}
  \widehat u( \xi, \eta ) =
  \begin{cases} 
    D^2\, \widehat v_{mn}, & (\xi,\eta) = (\xi_m,\eta_n) \\
    0,                     &\text{otherwise}
  \end{cases} 
\end{equation}
where $\xi_m = m/D$, $\eta_n = n/D$, for $m,n=-N/2,\dotsc,N/2$.  Using this
fact, the ASM approximates the inverse Fourier transform in
equation~\eqref{eq:ASMInvFourier} by the composite trapezoidal rule sampled at
the nodes $(m/D,n/D)$ so that
\begin{align}
  U(X,Y) &= e^{ikz} \dblintinf e^{i2\pi (\xi X + \eta Y)} 
            \exp \left(-i\frac{2\pi^2z}k (\xi^2 + \eta^2) \right) 
            \widehat u(\xi,\eta)\,d\xi\,d\eta \\
   &\approx e^{ikz} \frac 1{D^2} \sum_{m=-N/2}^{N/2} \sum_{n=-N/2}^{N/2} 
      e^{i2\pi (\xi_m X + \eta_n Y)} \exp \left(-i\frac{2\pi^2z}k (\xi_m^2 + \eta_n^2) \right) 
      \widehat u(\xi_m,\eta_n) \label{eq:Trapezoidal} \\
   &= e^{ikz} \sum_{m=-N/2}^{N/2} \sum_{n=-N/2}^{N/2} \widehat v_{mn} 
      e^{i2\pi (\xi_m X + \eta_n Y)} \exp \left(-i\frac{2\pi^2z}k (\xi_m^2 + \eta_n^2) \right),
\end{align}
which is the solution given in equation~\eqref{eq:ParaxialSolution}.

\medskip

We have seen that the continuous version of the Fresnel diffraction integral
has the effect of dispersing waves off to infinity in the transverse direction,
and the speed of these waves is proportional to their frequency. Unlike the
continuous version, however, the ASM imposes periodic boundary conditions; no
waves are dispersed to infinity, instead they are circled back into the domain.
Since the diffraction integrals eventually disperse all frequencies to
infinity, the error of the ASM grows with propagation distance. 

Perhaps more concerning is that the ASM violates the properties of the
convolution given in Corollary~\ref{cor:Bandwidth}. Even if the initial optical
field $u$ is bandlimited, the Fourier series approximation ignores this
property in favor of approximating it by a function of bounded support. This
function is no longer bandlimited, but its Fourier transform is again
approximated by a Fourier series. The diffraction integral turns this Fourier
series into a function with unbounded support---which is again truncated by a
Fourier series approximation. 

\subsection{The Rayleigh-Sommerfeld case}
We will again assume that $U$ takes the form $U(x,y,z) = e^{ikz} \psi(x,y,z)$.
For the Rayleigh-Sommerfeld integral, it can be shown that $\psi$ is the
solution of the following pseudodifferential Helmholtz propagator
equation~\cite{keefe_when_2018},
\begin{subequations}\label{eq:HelmholtzPropagator}
\begin{alignat}{2} 
   \frac{\del \psi}{\del z} &= \left(-ik + i\sqrt{k^2 + \Delta_{\perp}}\right) \psi, 
   &\qquad &(x,y,z) \in \R^2 \times (0,\infty), \label{eq:HelmholtzPropEqn} \\
   \psi(x,y,0) &= u(x,y), &&(x,y) \in \R^2, \label{eq:HelmholtzPropIC} 
\end{alignat} 
\end{subequations} 
The pseudodifferential operator is defined through the Fourier transform as
\begin{equation}
  \sqrt{k^2 + \Delta_{\perp}} \psi(x,y) \coloneqq \dblintinf 
                                        \sqrt{k^2 - 4 \pi^2 (\xi^2 + \eta^2)} \, 
                  \widehat \psi(\xi,\eta) \, e^{i2\pi(x\xi + y\eta)} d\xi\,d\eta,
\end{equation}
where the square root under the integral is defined with a branch cut such that
it is positive-real on the positive real axis and positive-imaginary on the
negative real axis.

As before, by taking $\psi$ to be a single Fourier mode, we obtain the
dispersion relation
\begin{equation}
  \omega = \sqrt{\frac{k^2}{4\pi^2} - \xi^2 - \eta^2} - \frac k{2\pi},
\end{equation}
and the group velocity
\begin{equation}
  v_g = |\nabla_{\xi,\eta}\omega| = \frac{\sqrt{\xi^2 + \eta^2}}
                                         {\sqrt{\frac{k^2}{4\pi^2} - \xi^2 - \eta^2}}.
\end{equation}

The Rayleigh-Sommerfeld case is more interesting in that the group velocity
increases without bound up to the \emph{finite} wavenumber $\xi^2 + \eta^2 =
k^2/(4\pi^2)$, at which point the group velocity is infinite. Whereas the group
velocity in the Fresnel case may still be small when $k$ is large compared to
spatial frequency content, controlling the group velocity in the Rayleigh case
is more difficult due to this singular behavior.

It can also be shown that ASM applied to the Rayleigh-Sommerfeld integral is
equivalent to solving the Helmholtz propagator equation with periodic boundary
conditions. As noted above, this situation is more concerning, since transverse
waves travel faster, with unbounded velocity in a finite range of frequencies.

\section{Diffraction integral approximation using sinc series} \label{sec:SincSeries}

As we have seen, the ASM is equivalent to imposing periodic boundary conditions on
the wave propagation problem, which leads to erroneous recycling of waves into
the domain which should be dispersed off to infinity.  The fundamental problem
is the attempt to represent the optical field by a function (a Fourier series)
with bounded support, which by Corollary~\ref{cor:Bandwidth} is transformed
into a function with unbounded support and is therefore no longer representable
using the same approximation.

Bandwidth, however, is preserved by the diffraction integrals, so we might
expect that the problems of the ASM could be eliminated by instead discretizing
the problem using a more suitable set of basis functions---one that is
naturally suited to problems on an infinite domain and that can approximate
bandlimited functions well. One such basis that satisfies these properties are
sinc functions, also known as the Whittaker cardinal basis. The key to such a
representation is the famous Shannon-Whittaker sampling theorem:
\begin{theorem}
Let $f(x)$ be a bandlimited function with bandwidth $W$. Then $f$ can be
represented exactly by the series
\begin{equation} 
  f(x) = \sum_{n=-\infty}^{\infty} f_n \sinc \left( \frac{x - x_n}{\Delta x} \right),
\end{equation} 
where $\Delta x = \frac 1{2W}$, $x_n = n\Delta x$, $f_n = f(x_n)$, and
$\sinc(x) = \sin(\pi x) / (\pi x)$.
\end{theorem}

The proof of this theorem relies on the assumption that the Fourier transform
of a bandlimited function $f$ can be represented by a Fourier series on the
interval $[-W,W]$:
\begin{equation}
  \widehat f(\xi) = \sum_{n=-\infty}^{\infty} c_n e^{-i \frac{\pi n}W \xi}, \quad \xi \in [-W,W].
\end{equation}
Taking the Fourier transform, we have
\begin{align}
  f(x) &= \intinf \widehat f(\xi) e^{i2\pi x \xi} \\
       &= \sum_{n=-\infty}^{\infty} c_n \intW e^{i\pi\left( 2x - \frac nW \right) \xi} d\xi \\
       &= \sum_{n=-\infty}^{\infty} c_n 2W \sinc( 2Wx - n ).
\end{align}
It follows that the Fourier coefficients of $\widehat f(\xi)$ and the samples
$f_n = f(x_n)$ are related by
\begin{equation}
  c_n = \frac{f_n}{2W} = \Delta x\, f_n.
\end{equation}

Of course, for computational purposes we cannot hope to represent a function
with an infinite number of samples. Therefore we truncate the sinc series as
\begin{equation} 
  f(x) \approx \sum_{n=-N/2}^{N/2} f_n \sinc \left( \frac{x - x_n}{\Delta x} \right),
\end{equation} 
for some even positive integer $N$. We will call such a function a
\emph{bandlimited function of order $N$}.

\subsection{Quadrature rules for diffraction integrals} \label{sec:QuadratureRules}

In this section we propose a method for obtaining quadrature rules that are
exact on the observation grid for bandlimited functions of order equal to the
number of sample points.

The method starts by approximating $u$ as a bandlimited function of bandwidth $W$
and order $N$. We assume for simplicity that the bandwidth and order are the
same in each dimension, with the spatial discretization given by $\delta =
1/(2W) = \Delta x = \Delta y$. Letting $x_j = j\delta$ and $y_{\ell} =
\ell\delta$ be the grid points on the source plane and $u_{j\ell} =
u(x_j,y_{\ell})$ be the values of $u$ at those points, we can write $u$ as
\begin{equation}
  u(x,y) = \sum_{j=-N/2}^{N/2} \sum_{\ell=-N/2}^{N/2} u_{j\ell} \phi_j(x) \phi_\ell(y),
\end{equation}
where we have written the sinc cardinal functions concisely as
\begin{equation}
  \phi_j(x) = \sinc \left( \frac{x - x_j}{\delta} \right), \quad\quad \phi_{\ell}(y) = \sinc \left( \frac{y - y_{\ell}}{\delta} \right).
\end{equation}
Then the diffraction integral is given by
\begin{equation} \label{eq:SincSolution}
  U(X,Y) = (h \ast u)(X,Y) = \sum_{j=-N/2}^{N/2} \sum_{\ell=-N/2}^{N/2} u_{j\ell} (h \ast (\phi_j \phi_{\ell}))(X,Y).
\end{equation}
The solution is simply the weighted sum of the diffraction integrals of the
cardinal functions. We may write it in terms of the diffraction integral of a
single basis function centered at the origin. Let
\begin{equation} \label{eq:PhiConvolution}
  \Phi(X,Y) = \dblintinf h(X-x,Y-y) \sinc \left( \frac x{\delta} \right) \sinc \left( \frac y{\delta} \right) dx\,dy.
\end{equation}
Using the convolution theorem and the fact that the Fourier transform of a sinc
function is given by
\begin{equation}
  \mathcal F \left[ \sinc \left( \frac x{\delta} \right) \right](\xi) = \delta \rect(\xi \delta), \quad\quad \rect(x) =
    \begin{cases}
      1, & -1/2 \leq x \leq 1/2, \\
      0, & \textrm{otherwise},
    \end{cases}
\end{equation}
we can also write $\Phi(X,Y)$ as an inverse Fourier transform,
\begin{equation} \label{eq:PhiFourier}
  \Phi(X,Y) = \delta^2 \dblintW \widehat h(\xi,\eta) e^{i 2\pi (X \xi + Y \eta)} d\xi\,d\eta,
\end{equation}
where the finite integral comes from the fact that the Fourier transform of the
cardinal function is zero outside of the square $[-W,W]^2$. From this form it is
clear that $\Phi(X,Y)$ has bandwidth $W$, thus preserving the bandwidth of the
cardinal function. It follows that
\begin{equation}
  (h \ast (\phi_j \phi_{\ell}))(X,Y) = \Phi(X-x_j,Y-y_{\ell}),
\end{equation}
and we can write the solution~\eqref{eq:SincSolution} as
\begin{equation}
  U(X,Y) = \sum_{j=-N/2}^{N/2} \sum_{\ell=-N/2}^{N/2} u_{j\ell} \Phi(X-x_j,Y-y_{\ell}).
\end{equation}

Now suppose we wish to evaluate the diffraction integral above at a set of
observation points $(X_m,Y_n)$, which need not be the same as the source grid.
Denoting the diffraction quadrature weights by
\begin{equation} \label{eq:DiffractionWeights}
  w_{mj;\,n\ell} = \Phi(X_m-x_j,Y_n-y_{\ell}),
\end{equation}
it follows that the diffraction integral for an order $N$ bandlimited function at
the points $(X_m,Y_n)$ are given \emph{exactly} by
\begin{equation} \label{eq:DiffractionQuad}
  U(X_m,Y_n) = \sum_{j=-N/2}^{N/2} \sum_{\ell=-N/2}^{N/2} w_{mj;\,n\ell} u_{j\ell}.
\end{equation}
To distinguish between the diffraction integrals, we denote the Fresnel weights
by $w^F_{mj;n\ell}$ and the Rayleigh-Sommerfeld weights by $w^{RS}_{mj;n\ell}$.
Since there is no closed-form solution to the Rayleigh-Sommerfeld weights,
these can be computed with any high-order quadrature, such as Fej\'{e}r's first
quadrature rule~\cite{Waldvogel2006} that we use in the numerical results
section.  On the other hand, the Fresnel weights do have a closed-form analytical
expression, which will be derived in the following subsection. 

\subsubsection{The Fresnel weights} \label{sec:FresnelWeights}

Because the Fresnel kernel is the product of one-dimensional functions, we may
write the function $\Phi(X,Y)$ as the product of one-dimensional weights:
\begin{equation} \label{eq:PhiProduct}
  \Phi(X,Y) = e^{ikz} \varphi(X) \varphi(Y),
\end{equation}
where the function $\varphi$ can be expressed in two ways:
\begin{align}
  \varphi(X) &= e^{-i\frac{\pi}4} \sqrt{\frac k{2\pi z}} \intinf e^{i \frac k{2z} (X-x)^2} 
                \sinc \left( \frac x{\delta} \right) \,dx \label{eq:PhiPhysical}, \\
  \varphi(X) &= \delta \intW \exp\left( -i\frac{2\pi^2z}k \xi^2 \right) e^{i 2\pi X \xi}  d\xi. 
                \label{eq:LittlePhiFourier}
\end{align}
The one-dimensional function $\varphi$ has an analytical expression in terms of
the Fresnel cosine and sine integral functions $C(x)$ and $S(x)$ which are defined as
\begin{equation}
  C(x) = \int_0^x \cos(\mu^2)\,d\mu, \quad\quad S(x) = \int_0^x \sin(\mu^2)\,d\mu.
\end{equation}
Using the change of variable $\mu = \pi\sqrt{(2z)/k} \, \xi - \sqrt{k/(2z)} \,
X$, we have
\begin{align}
  \varphi(X) &= \delta \intW \exp \left( -i\frac{2\pi^2z}k \xi^2 \right) e^{i 2\pi X \xi} d\xi \\
          &= \delta \intW \exp \left[ -i\frac{2\pi^2z}k \left(\xi -  \frac{X k}{2\pi z} \right)^2 + i\frac{X^2 k}{2z} \right] d\xi \\
          &= \frac{\delta}{\pi} \sqrt{\frac k{2z}} \exp \left( i\frac{X^2 k}{2z} \right) \int_{\mu_1}^{\mu_2} \exp( -i\mu^2 )\, d\mu \\
          &= \frac{\delta}{\pi} \sqrt{\frac k{2z}} \exp \left( i\frac{X^2 k}{2z} \right) \bigg( C(\mu_2) - C(\mu_1) - i \big( S(\mu_2) - S(\mu_1) \big) \bigg),
\end{align}
where
\begin{equation}
  \mu_1 = -\pi\sqrt{\frac{2z}k} W - \sqrt{\frac k{2z}} X, \quad\quad \mu_2 = \pi\sqrt{\frac{2z}k} W - \sqrt{\frac k{2z}} X.
\end{equation}
It follows that we may write the Fresnel quadrature weights as
\begin{equation}\label{eq:FresnelWeights}
  w^F_{mj;\,n\ell} = e^{ikz} \omega^x_{mj} \omega^y_{n\ell},
\end{equation}
where
\begin{equation}
  \omega^x_{mj} = \varphi(X_m - x_j) \quad\text{ and }\quad 
  \omega^y_{n\ell} = \varphi(Y_n - y_{\ell}).
\end{equation}
%

\subsection{Computing the diffraction quadratures} \label{sec:ComputingQuadrature}

We note that equation~\eqref{eq:DiffractionQuad} can be written as a ``scaled
discrete convolution,''~\cite{nascov_fast_2009} which can be computed with FFT
speed by using an algorithm based on the fractional Fourier
transform~\cite{bailey_fractional_1991}.

However, if the source and observation grids use the same spatial
discretization $\delta$, then equation~\eqref{eq:DiffractionQuad} becomes a
discrete convolution, where the weights can be written as the difference
between corresponding indices in each dimension---i.e.,
\begin{align}
  w_{mj;n\ell} &= w_{m-j,n-\ell} \nonumber \\
    &= \delta^2 \dblintW \widehat h(\xi,\eta) e^{i 2\pi \delta ((m - j) \xi + (n - \ell) \eta)} d\xi\,d\eta.
\end{align}
The resulting discrete convolution can be computed directly using zero-padded
FFTs of size $2N$. In this case, the sinc-based algorithm has the same
computational complexity as the ASM. 

In the case of Fresnel diffraction, the algorithm can also be written as two
matrix multiplications. Let $\mathbf u$, $\mathbf U$, $\bome^x$, and $\bome^y$
denote matrices whose entries are given by
\begin{alignat}{2}
  \mathbf u_{j\ell} &= u(x_j,y_{\ell}), &\qquad \mathbf U_{mn} &= U(X_m,Y_n), \nonumber \\
  \bome^x_{mj} &= \ome^x_{mj},          & \bome^y_{n\ell} &= \ome^y_{n\ell}. \label{eq:OmegaMatrices}
\end{alignat}
Then equation~\eqref{eq:DiffractionQuad} can be written compactly as
\begin{equation}\label{eq:FresnelMxM}
  \mathbf U = e^{ikz} \boldsymbol \omega^x \, \mathbf u \, \left( \boldsymbol \omega^y \right)^T,
\end{equation}
where the supersript $T$ denotes the matrix transpose (no conjugation!).
Although the FFT has a faster computational complexity than matrix
multiplication, in practice it is only faster for large enough grids. As is,
the above Fresnel algorithm encapsulated in~\eqref{eq:FresnelMxM} is simpler
and faster for small grids. In the next section, we will see how the algorithm
can be even faster for large Fresnel numbers, where the algorithm reduces to
banded matrix multiplication.

\subsection{Further truncation of the quadrature weights} \label{sec:KernelTrunc}

In the previous section we showed that the diffraction quadrature formulas lead
to a fast algorithm, since the discrete convolution can be computed in
$\mathcal O(N\log N)$ time using the fast Fourier transform. We emphasize that
the only approximation involved is in the optical field $u$; we assume that it
can be accurately represented by a bandlimited function of order $N$.  Given
such an approximation, the subsequent discrete quadrature formula is equal to
the continuous convolution at the observation points, regardless of wavelength
or propagation distance; there is no approximation in this step.

In this section, however, we show how the weights can be further truncated to
just a few terms when the Fresnel number is large (i.e., when $(D^2/(\lambda
z))$ is large, where $D$ is the spatial extent of the optical field) and when
the observation grid is the same as the source grid. The result is an even
faster $O(N)$ algorithm for computing the convolution for any fixed error
tolerance. Recalling the method of stationary phase, the highly oscillatory
kernels in the large Fresnel number regime lead us to expect that most of the
contribution to the diffraction integrals will come from a small neighborhood
of the point $(X,Y)$. The following lemma rigorously establishes this argument.

\begin{lemma} \label{lem:Asymptotic}
The function $\varphi(X)$ given by equation~\eqref{eq:LittlePhiFourier} has the
following asymptotic behavior as $Xk\delta/(\pi z) \rightarrow \infty$:
\begin{align}
  \varphi(X) = \frac{\delta}{\pi X} \Bigg[ \exp \left( -i \frac{\pi^2 z}{2 k \delta^2} \right) \frac 1{1 - i\,z/(kX^2)} \left( \sin \big(\pi X/\delta\big) - i \frac{\pi z}{Xk\delta} \cos(\pi X/\delta) \right) + \mathcal O \Big( \big( \pi z/(Xk\delta) \big)^2 \Big) \Bigg]
\end{align}
\end{lemma}

\begin{proof}
Using the change of variable $\mu = \xi/W$ and letting $A = \pi X/\delta$ and
$\kappa = \pi^2 z/(k\delta^2)$, we can write
equation~\eqref{eq:LittlePhiFourier} as
\begin{equation}
  \phi(X) = \frac 12 \int_{-1}^1 \exp \big( -i (\kappa/2)\mu^2 \big) e^{iA\mu} \, d\mu.
\end{equation}
Integrating by parts twice, letting $\varepsilon = \pi z/(Xk\delta)$ (which by
assumption is going to zero), we have
\begin{align}
  \phi(X) &= \left[ \frac{e^{iA\mu}}{2iA} \exp \big( -i (\kappa/2)\mu^2 \big) \right|_{-1}^1 - \frac 1{2iA} \int_{-1}^1 (-i\kappa\mu) \exp \big( -i (\kappa/2)\mu^2 \big) e^{iA\mu} \, d\mu \\
    &= \frac{\sin A}{A} e^{-i\kappa/2} + \frac{\varepsilon}2 \int_{-1}^1 \mu \exp \big( -i (\kappa/2)\mu^2 \big) e^{iA\mu} \, d\mu \\
    &= \frac{\sin A}{A} e^{-i\kappa/2} + \frac{\varepsilon}2 \left[ \mu \exp \big( -i (\kappa/2)\mu^2 \big) e^{iA\mu} \right|_{-1}^1 - \frac{\varepsilon}{2iA} \int_{-1}^1 (1-i\kappa\mu^2) \exp \big( -i (\kappa/2)\mu^2 \big) e^{iA\mu} \, d\mu \\
    &= \frac{e^{-i\kappa/2}}{A} \big( \sin A - i \varepsilon \cos A \big) + i\frac{\varepsilon}{A} \phi(X) + \frac{\varepsilon^2}2 \int_{-1}^1 \mu^2 \exp \big( -i (\kappa/2)\mu^2 \big) e^{iA\mu} \, d\mu.
\end{align}
By integrating by parts again, we would see that the last integral above is a
term of order $\varepsilon^2/A$. Therefore, dropping that integral, rearranging
terms, and substituting the definitions of $\varepsilon$, $\kappa$, and $A$
leads to the desired result.
\end{proof}

The following result is a straightforward application of the above proposition
and equation~\eqref{eq:PhiProduct}, along with the fact that the Fresnel
diffraction integral is an asymptotic expression of the Rayleigh-Sommerfeld
integral in the specified regime.  

\begin{proposition}\label{prop:AsympWeights}
At the points $(X_m,Y_n) = (m\delta,n\delta)$ for integers $m$ and
$n$, the Rayleigh-Sommerfeld and Fresnel weights have the following asymptotic
behavior as $X_m^2 k/z \rightarrow \infty$ and $Y_n^2 k/z \rightarrow
\infty$.
\begin{align}
  \Phi(X_m,0)   &= -i (-1)^m \left( \frac z{kX_m^2} \right) e^{ikz} \exp \left( -i \frac{\pi^2 z}{2k\delta^2} \right) + \mathcal O \big( (X_m^2 k/z)^{-2} \big), \quad m \neq 0, \\
  \Phi(0,Y_n)   &= -i (-1)^n \left( \frac z{kY_n^2} \right) e^{ikz} \exp \left( -i \frac{\pi^2 z}{2k\delta^2} \right) + \mathcal O \big( (Y_n^2 k/z)^{-2} \big), \quad n \neq 0, \label{eq:AsymptoticWeight2} \\
  \Phi(X_m,Y_n) &=  \mathcal O \big[ (z/k)^2 \big( X_m^{-4} + X_m^{-2} Y_n^{-2} + Y_n^{-4} \big) \big], \quad m,n \neq 0. \label{eq:AsymptoticWeight3}
\end{align}
\end{proposition}

These asymptotic expressions for the Fresnel quadrature weights can be used
to estimate a cutoff number $M$ such that the weights $w_{m-j,n-\ell} =
\Phi(X_{m-j},Y_{n-\ell})$ for $m-j$ and $n-\ell$ greater than $M$ can
effectively be set to zero subject to a prescribed error tolerance. The effect
of this truncation is to turn the matrices $\boldsymbol \omega^x$ and
$\boldsymbol \omega^y$ from equation~\eqref{eq:OmegaMatrices} into banded
matrices.

We note that the ASM can produce fairly accurate results when the Fresnel
number is large---i.e., when the propagation distance is small enough and the
domain large enough so that most of the transverse waves do not yet reach the
edge of the domain.  However, using the ASM in this case will still be less
efficient than the sinc method since it is precisely in this regime that
Proposition~\ref{prop:AsympWeights} applies and the banded matrix
multiplication algorithm for the sinc method will outperform even an FFT-based
algorithm.

\section{Numerical results} \label{sec:NumericalResults}

In this section, we perform a convergence analysis on a problem for which there
is an exact solution, Fresnel propagation of a Gaussian beam, and we compare
numerical results obtained using the ASM and the sinc method. We also compare
the methods qualitatively for Fresnel and Rayleigh-Sommerfeld diffraction of
the circular aperture problem.

\subsection{Gaussian beam} \label{sec:GaussianBeam}
%
\begin{figure}[htb!]
  \begin{center}
    \subfigure[]{\label{fig:GauBeam-1}
    \includegraphics[angle=0,width=\twoFigWidth\textwidth]{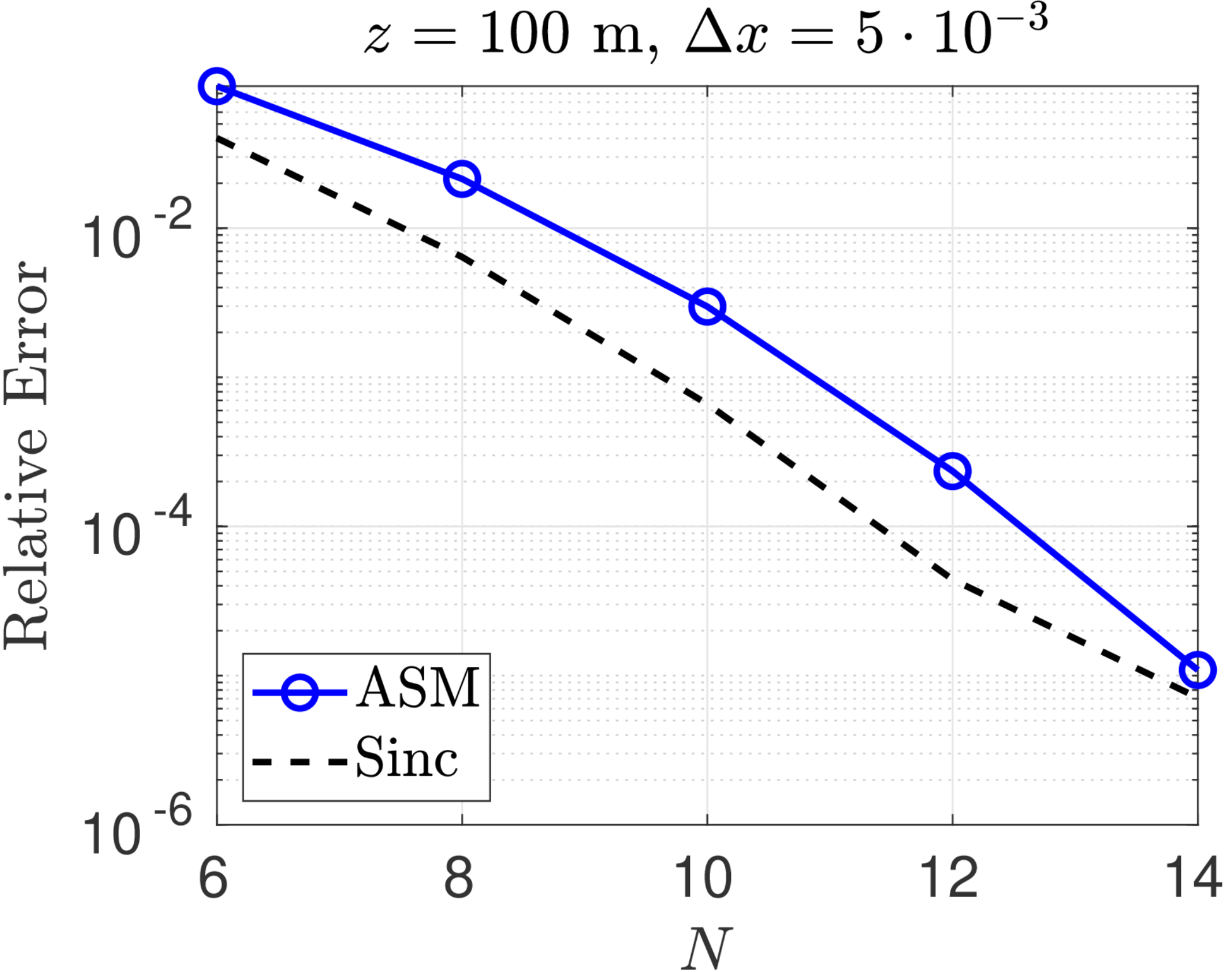}}
    \subfigure[]{\label{fig:GauBeam-2}
    \includegraphics[angle=0,width=\twoFigWidth\textwidth]{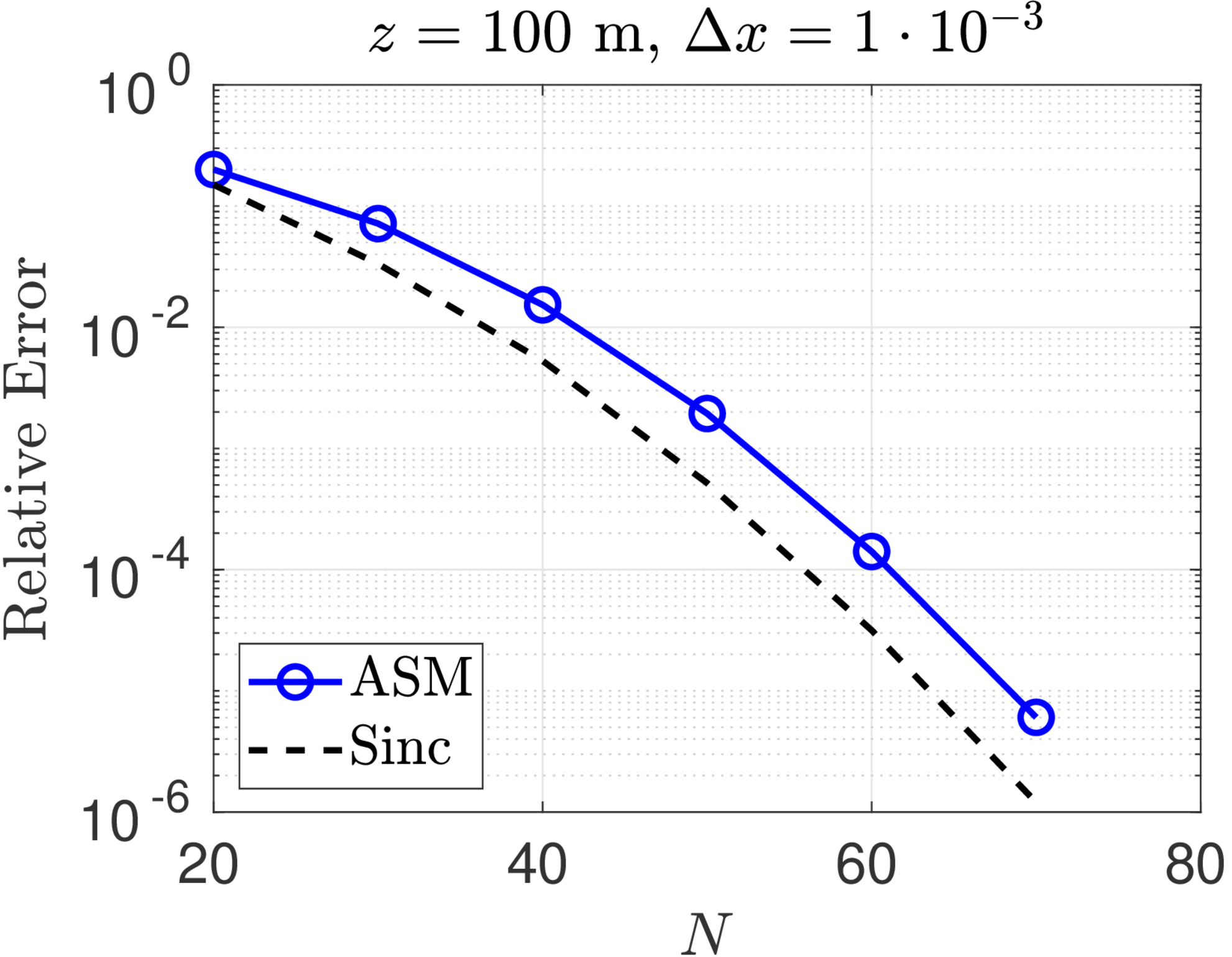}}
    \subfigure[]{\label{fig:GauBeam-3}
    \includegraphics[angle=0,width=\twoFigWidth\textwidth]{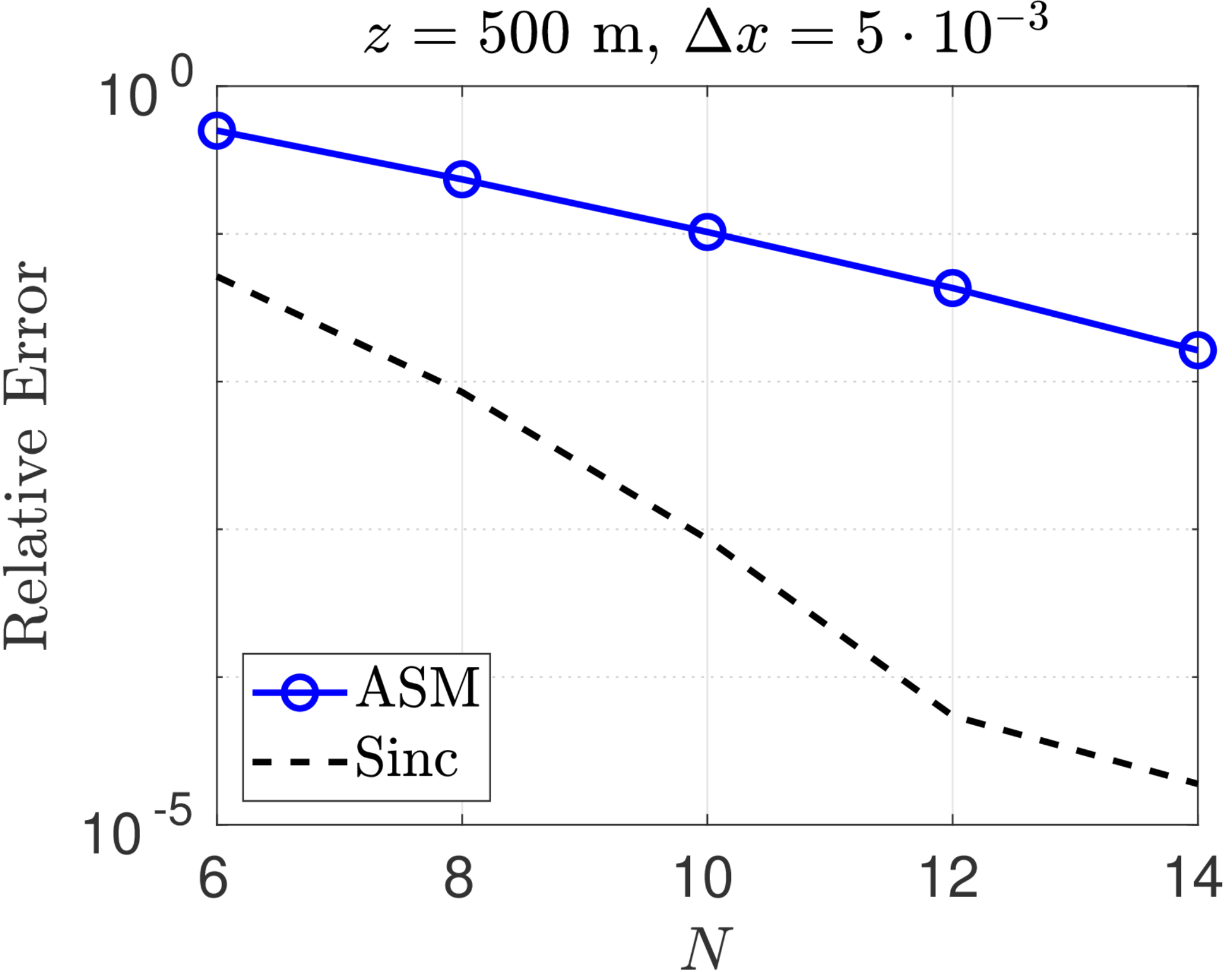}}
    \subfigure[]{\label{fig:GauBeam-4}
    \includegraphics[angle=0,width=\twoFigWidth\textwidth]{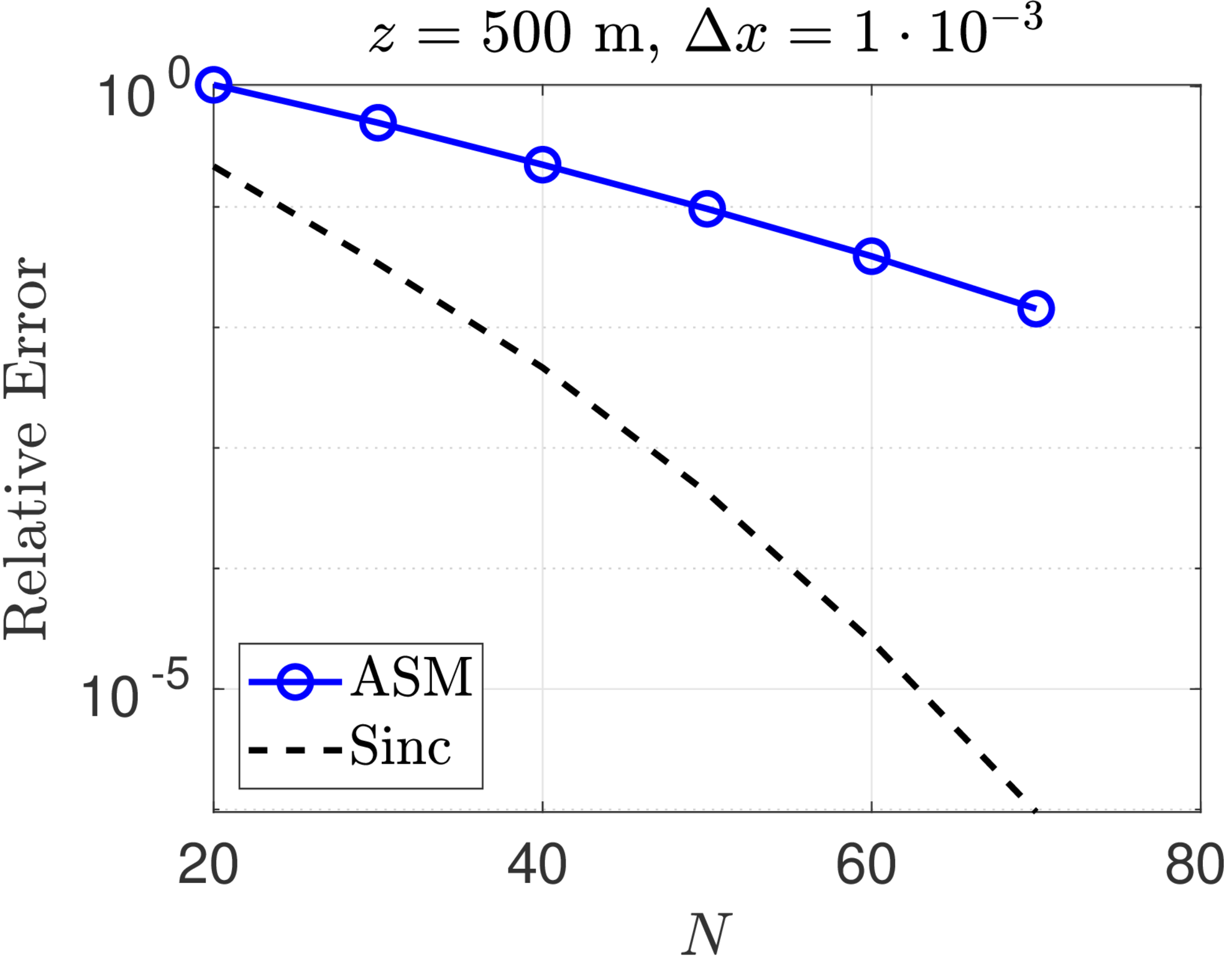}}
    \subfigure[]{\label{fig:GauBeam-5}
    \includegraphics[angle=0,width=\twoFigWidth\textwidth]{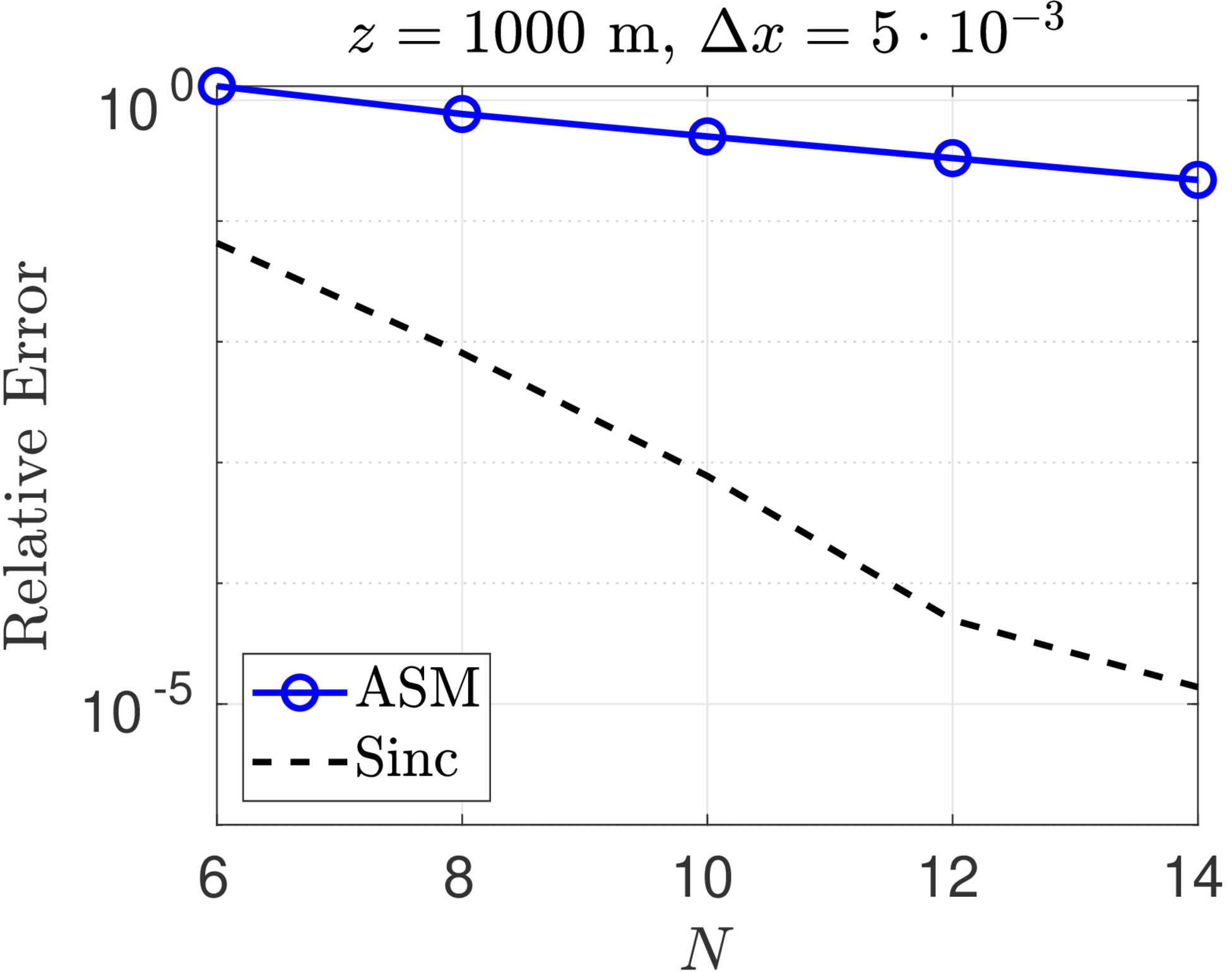}}
    \subfigure[]{\label{fig:GauBeam-6}
    \includegraphics[angle=0,width=\twoFigWidth\textwidth]{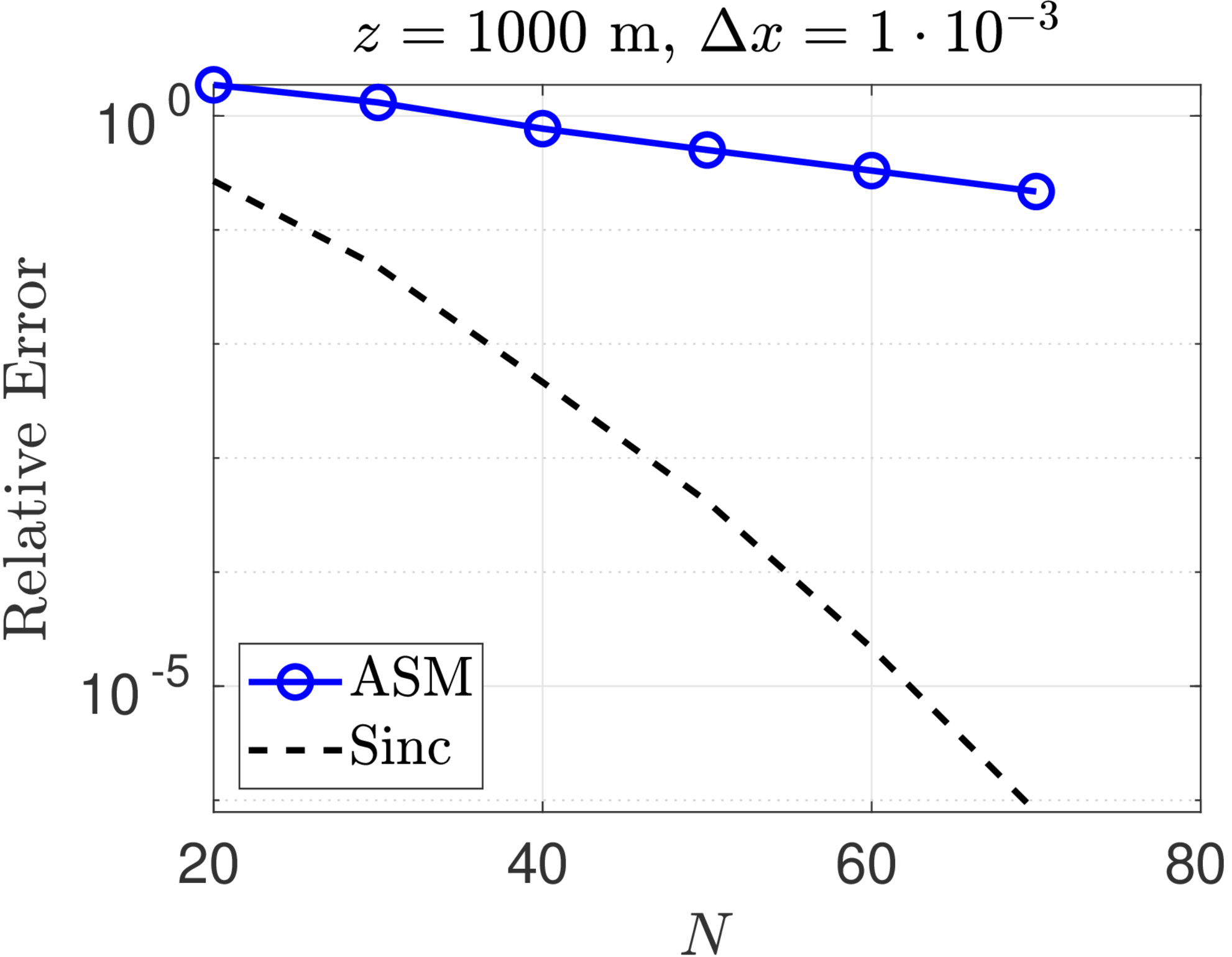}}
  \end{center}
  \caption{Convergence of the ASM and sinc method for Fresnel diffraction of a
           Gaussian beam. The plots show relative errors versus number of points on one
           side of an $N \times N$ grid. The propagation distance $z$ and spatial sample
           spacing $\Delta x$ are indicated above each plot.} \label{fig:GauBeamRelErr}
\end{figure}
We consider the optical propagation of a Gaussian beam with unit amplitude and
initial waist radius $w_0$,
\begin{equation}
  u(x,y) = \exp \left( -\frac{x^2 + y^2}{w_0^2} \right).
\end{equation}
Its Fourier transform is given by
\begin{equation}
  \widehat u(\xi,\eta) = \pi w_0^2 \exp \left( -\pi^2 w_0^2 (\xi^2 + \eta^2) \right),
\end{equation}
and the exact solution to the Fresnel diffraction integral is given by the
formula
\begin{align}
  U(X,Y) &= \frac 1{\sqrt{1 + \Big( \frac{2z}{kw_0^2} \Big)^2}} 
    \exp \left\{ -\frac 1{w_0^2}\, \frac{X^2 + Y^2}{1 + \Big( \frac{2z}{kw_0^2} \Big)^2}
      + ikz - i \arctan \left( \frac{2z}{kw_0^2} \right)
      + i \frac{k}{2z}\, \frac{X^2 + Y^2}{1 + \Big( \frac{kw_0^2}{2z} \Big)^2} \right\}.
\end{align}
Figure~\ref{fig:GauBeamRelErr} shows convergence plots for the sinc method and
the ASM for Fresnel diffraction of a Gaussian beam with wavelength $\lambda =
1\,\mu\textrm{m}$ and beam radius $w_0 = 1\,\textrm{cm}$. Three different
propagation distances, $z\in \{100,\, 500, \, 1000\}$ m, and two sample
spacings, $\Delta x \in \{ 1\cdot 10^{-3},\, 5\cdot 10^{-3}\}$ m, were chosen
for this study. The exact solution is computed at every point of an observation
plane grid of $N\times N$ points and the $2$-norm relative error is computed
for both the ASM and the sinc numerical solutions.  The errors for the ASM are
competitive with (but still greater than) the sinc method at shorter
propagation distances. As the analysis predicted, however, ASM performs worse
at longer distances. On the other hand, for fixed $\Delta x$, the convergence
of the sinc method is almost identical for all propagation distances.

\subsection{Circular aperture} \label{sec:CircAper}

Next, we consider a diffraction problem for a circular aperture using both the
Fresnel and Rayleigh-Sommerfeld integrals. Unfortunately, in this case, there
is no closed-form solution, even for the simpler Fresnel case.

The problem setup is as follows (all units are in meters). The field in the source
plane is given by
\begin{equation}\label{eq:SphWave}
 u(x,y,z) = A \cfun\left( \frac{r}{r_0} \right) 
          \frac{ e^{ik\sqrt{ (x-x_0)^2 + (y-y_0)^2 +(z-z_0)^2 } } }
                       { \sqrt{ (x-x_0)^2 + (y-y_0)^2 +(z-z_0)^2 } }, \qquad
  \cfun( r ) \coloneqq
    \begin{cases}
      1, &r \leq 1, \\
      0, & \textrm{otherwise},
    \end{cases}
\end{equation}
where the wavelength $\lambda = 10\,\mu\textrm{m}$, $r = \sqrt{x^2+y^2}$, $r_0 = 10^{-3}$ is the circular aperture radius, $A = 3\cdot 10^{-2}$ is the source amplitude, and $(x_0,y_0,z_0) = (0,0,-3\cdot
10^{-2})$ so that the source is placed slightly behind the an opaque screen at
$z=0$. The source $u$ is evaluated over a planar grid with $N\times N$ points
for $(x,y) \in [-2\cdot 10^{-3}, 2\cdot 10^{-3}]^2$, with $N = 400$ in all
cases, which we then propagate using each of the diffraction
integrals~\eqref{eq:FresnelIntegral} and~\eqref{eq:RayleighSommerfeld}.

Figure~\ref{fig:CircAperIrrad} shows irradiance patterns ($|U|^2$) for Fresnel
and Rayleigh-Sommerfeld diffraction, computed with the ASM and sinc methods,
for a circular aperture over a uniform $N\times N$ planar grid at $z = 0.015$.
Because the source field of the aperture problem is discontinuous, we expect
that both the Fourier series and sinc series approximations will suffer from
the Gibbs phenomenon. However, recalling the discussion of group velocity in
Section~\ref{sec:FresnelFourier}, the highest frequency modes are quickly
dispersed off to infinity so that the solution becomes smoother at longer
propagation distances. This behavior is evident in the solution given by the
sinc method; on the other hand, ASM retains all the high frequency
oscillations, which are clearly visible in the pseudocolor plots of both the
Fresnel and Rayleigh-Sommerfeld cases, regardless of propagation distance.
(Pseudocolor plots were generated using the visualization software
VisIt~\cite{HPV_VisIt}.)
\begin{figure}[htb!]
    \centering                                                               
    \includegraphics[width=0.8\textwidth]{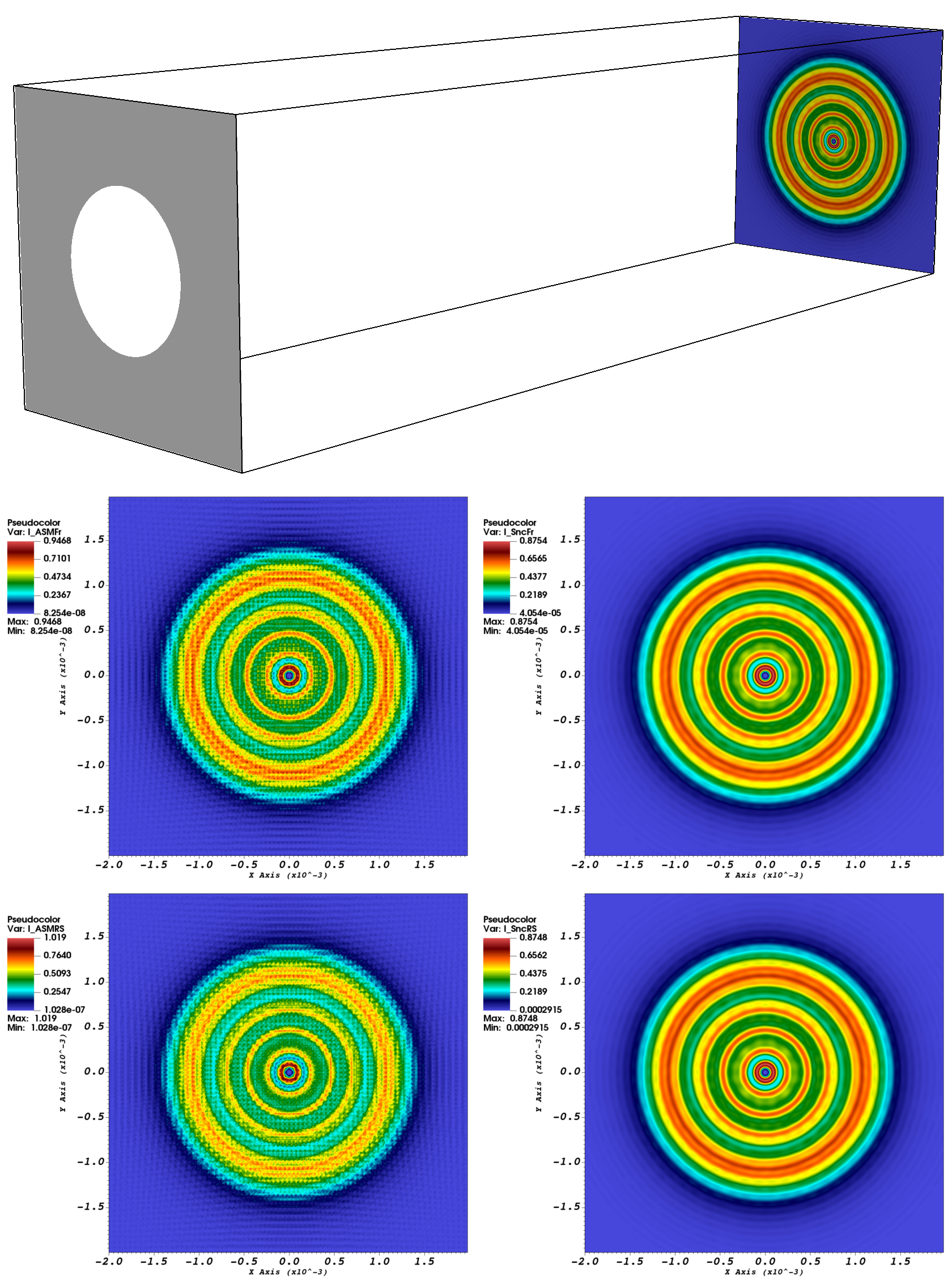}
    \caption{Irradiance plots ($|U|^2$) for Fresnel and Rayleigh-Sommerfeld (RS) diffraction
             of a circular aperture of radius $r_0 = 10^{-3}$ m at $z=0$ illuminated by a 
             spherical wave located at $(x_0,y_0,z_0) = (0,0,-3\cdot 10^{-2})$
             m. The Fresnel (\texttt{I\_ASMFr}) and RS (\texttt{I\_ASMRS}) ASM irradiance
             at $z = 0.015$ m are shown in the left column plots. The corresponding sinc results 
             (\texttt{I\_SncFr}) and (\texttt{I\_SncRS}) are displayed in the right column.} 
             \label{fig:CircAperIrrad}
\end{figure}         

\section{Conclusion}

We presented a method based on sinc approximations to compute
Rayleigh-Sommefeld (RS) and Fresnel diffraction integrals. We compared our
approach with the standard method used to evaluate these optical wave
propagation integrals, the angular spectrum method (ASM), and we showed that
the ASM introduces artificial periodicity into the domain. The ASM's artificial
periodicity, in turn, leads to errors in the optical fields which grow as the
propagation distance increases.  On the other hand, our sinc-based approach
does not impose artifical periodic boundary conditions and instead correctly
disperses transverse modes to infinity. Additionally, unlike the ASM, the sinc
method preserves the bandwidth of the underlying convolution and the accuracy
of the propagated field depends only on the approximation accuracy of the
source field and is independent of wavelength, propagation distance, and
observation plane discretization.  Like the ASM, our sinc method can also be
formulated as a discrete convolution over uniform grids and can thus be
computed using FFTs.  When the Fresnel number is large, only a few terms of the
sinc quadrature are necessary to achieve a prescribed error tolerance; in this
case, the sinc-based approach simplifies to an even faster $O(N)$ algorithm.
For the Fresnel case, the set of sinc integration weights was given
analytically and numerical results confirmed that the sinc approach achieves
high-order accuracy for both short-range and long-range propagation.  Similar
results were obtained with the sinc-based method for RS diffraction through a
circular aperture even though in that case the sinc integration weights were
computed numerically.  As expected, for a fixed discretization resolution, the
ASM error confirmed that the solution deteriorates with increasing propagation
distance.

\section*{Acknowledgements}

This work was supported by the Air Force Office of Scientific Research [grant number 20RDCOR016].

\section*{Disclosures}

Approved for public release; distribution is unlimited. Public Affairs release
approval number AFRL-2021-3954.

\bibliography{refs}

\begin{thebibliography}{18}
\expandafter\ifx\csname natexlab\endcsname\relax\def\natexlab#1{#1}\fi
\providecommand{\url}[1]{\texttt{#1}}
\providecommand{\href}[2]{#2}
\providecommand{\path}[1]{#1}
\providecommand{\DOIprefix}{doi:}
\providecommand{\ArXivprefix}{arXiv:}
\providecommand{\URLprefix}{URL: }
\providecommand{\Pubmedprefix}{pmid:}
\providecommand{\doi}[1]{\href{http://dx.doi.org/#1}{\path{#1}}}
\providecommand{\Pubmed}[1]{\href{pmid:#1}{\path{#1}}}
\providecommand{\bibinfo}[2]{#2}
\ifx\xfnm\relax \def\xfnm[#1]{\unskip,\space#1}\fi
\bibitem[{Sherman(1967)}]{Sherman1967Int}
\bibinfo{author}{G.~C. Sherman},
\newblock \bibinfo{title}{Integral-transform formulation of diffraction
  theory},
\newblock \bibinfo{journal}{JOSA} \bibinfo{volume}{57} (\bibinfo{year}{1967})
  \bibinfo{pages}{1490--1498}.
\bibitem[{Shewell and Wolf(1968)}]{ShewellWolf1968}
\bibinfo{author}{J.~Shewell}, \bibinfo{author}{E.~Wolf},
\newblock \bibinfo{title}{Inverse diffraction and a new reciprocity theorem},
\newblock \bibinfo{journal}{JOSA} \bibinfo{volume}{58} (\bibinfo{year}{1968})
  \bibinfo{pages}{1596--1603}.
\bibitem[{Goodman(1996)}]{Goodman1996}
\bibinfo{author}{J.~Goodman}, \bibinfo{title}{Introduction to Fourier Optics},
  Electrical Engineering Series, \bibinfo{publisher}{McGraw-Hill},
  \bibinfo{year}{1996}.
\bibitem[{Schmidt(2010)}]{schmidt_numerical_2010}
\bibinfo{author}{J.~D. Schmidt}, \bibinfo{title}{Numerical {Simulation} of
  {Optical} {Wave} {Propagation} with {Examples} in {MATLAB}},
  \bibinfo{publisher}{SPIE}, \bibinfo{year}{2010}.
\bibitem[{Voelz(2011)}]{voelz_computational_2011}
\bibinfo{author}{D.~G. Voelz}, \bibinfo{title}{Computational {F}ourier
  {O}ptics: a {MATLAB} {T}utorial}, \bibinfo{publisher}{SPIE},
  \bibinfo{year}{2011}.
\bibitem[{Matsushima and Shimobaba(2009)}]{matsushima_band-limited_2009}
\bibinfo{author}{K.~Matsushima}, \bibinfo{author}{T.~Shimobaba},
\newblock \bibinfo{title}{Band-{Limited} {Angular} {Spectrum} {Method} for
  {Numerical} {Simulation} of {Free}-{Space} {Propagation} in {Far} and {Near}
  {Fields}},
\newblock \bibinfo{journal}{Optics Express} \bibinfo{volume}{17}
  (\bibinfo{year}{2009}) \bibinfo{pages}{19662--19673}.
\bibitem[{Zhang et~al.(2020)Zhang, Zhang, Sheppard, and
  Jin}]{zhang_analysis_2020}
\bibinfo{author}{W.~Zhang}, \bibinfo{author}{H.~Zhang},
  \bibinfo{author}{C.~J.~R. Sheppard}, \bibinfo{author}{G.~Jin},
\newblock \bibinfo{title}{Analysis of numerical diffraction calculation
  methods: from the perspective of phase space optics and the sampling
  theorem},
\newblock \bibinfo{journal}{JOSA A} \bibinfo{volume}{37} (\bibinfo{year}{2020})
  \bibinfo{pages}{1748--1766}.
\bibitem[{Borel(1897)}]{borel_sur_1897}
\bibinfo{author}{E.~Borel},
\newblock \bibinfo{title}{Sur l{\textquoteright}interpolation},
\newblock \bibinfo{journal}{CR Acad. Sci. Paris} \bibinfo{volume}{124}
  (\bibinfo{year}{1897}) \bibinfo{pages}{673--676}.
\bibitem[{Whittaker(1915)}]{whittaker_functions_1915}
\bibinfo{author}{E.~T. Whittaker},
\newblock \bibinfo{title}{On the {Functions} which are represented by the
  {Expansions} of the {Interpolation}-{Theory}},
\newblock \bibinfo{journal}{Proceedings of the Royal Society of Edinburgh}
  \bibinfo{volume}{35} (\bibinfo{year}{1915}) \bibinfo{pages}{181--194}.
\bibitem[{Stenger(2012)}]{Stenger2012}
\bibinfo{author}{F.~Stenger}, \bibinfo{title}{Numerical methods based on Sinc
  and analytic functions}, volume~\bibinfo{volume}{20},
  \bibinfo{publisher}{Springer Science \& Business Media},
  \bibinfo{year}{2012}.
\bibitem[{Debnath and Mikusinski(2005)}]{DebnathMikusinskiBook}
\bibinfo{author}{L.~Debnath}, \bibinfo{author}{P.~Mikusinski},
  \bibinfo{title}{Introduction to Hilbert spaces with applications},
  \bibinfo{publisher}{Academic press}, \bibinfo{year}{2005}.
\bibitem[{Lalor(1968)}]{lalor_conditions_1968}
\bibinfo{author}{{\'E}.~Lalor},
\newblock \bibinfo{title}{Conditions for the validity of the angular spectrum
  of plane waves},
\newblock \bibinfo{journal}{Journal of the Optical Society of America}
  \bibinfo{volume}{58} (\bibinfo{year}{1968}) \bibinfo{pages}{1235--1237}.
\bibitem[{Blackledge(2005)}]{Blackledge2005digital}
\bibinfo{author}{J.~Blackledge}, \bibinfo{title}{Digital Image Processing:
  Mathematical and Computational Methods}, Woodhead Publishing Series in
  Electronic and Optical Materials, \bibinfo{publisher}{Elsevier Science},
  \bibinfo{year}{2005}.
\bibitem[{Keefe et~al.(2018)Keefe, Zilberter, and Madden}]{keefe_when_2018}
\bibinfo{author}{L.~Keefe}, \bibinfo{author}{I.~Zilberter},
  \bibinfo{author}{T.~J. Madden},
\newblock \bibinfo{title}{When {Parabolized} {Propagation} {Fails}: {A}
  {Matrix} {Square} {Root} {Propagator} for {EM} {Waves}},
\newblock in: \bibinfo{booktitle}{Plasmadynamics and {Lasers} {Conference}},
  \bibinfo{publisher}{AIAA}, \bibinfo{year}{2018}.
  \DOIprefix\doi{10.2514/6.2018-3113}.
\bibitem[{Waldvogel(2006)}]{Waldvogel2006}
\bibinfo{author}{J.~Waldvogel},
\newblock \bibinfo{title}{Fast construction of the {F}ej{\'e}r and
  {C}lenshaw--{C}urtis quadrature rules},
\newblock \bibinfo{journal}{BIT Numerical Mathematics} \bibinfo{volume}{46}
  (\bibinfo{year}{2006}) \bibinfo{pages}{195--202}.
\bibitem[{Nascov and Logof{\u a}tu(2009)}]{nascov_fast_2009}
\bibinfo{author}{V.~Nascov}, \bibinfo{author}{P.~C. Logof{\u a}tu},
\newblock \bibinfo{title}{Fast computation algorithm for the
  {Rayleigh}-{Sommerfeld} diffraction formula using a type of scaled
  convolution},
\newblock \bibinfo{journal}{Applied Optics} \bibinfo{volume}{48}
  (\bibinfo{year}{2009}) \bibinfo{pages}{4310--4319}.
\bibitem[{Bailey and Swarztrauber(1991)}]{bailey_fractional_1991}
\bibinfo{author}{D.~H. Bailey}, \bibinfo{author}{P.~N. Swarztrauber},
\newblock \bibinfo{title}{The fractional {Fourier} transform and applications},
\newblock \bibinfo{journal}{SIAM Review} \bibinfo{volume}{33}
  (\bibinfo{year}{1991}) \bibinfo{pages}{389--404}.
\bibitem[{Childs et~al.(2012)Childs, Brugger, Whitlock, Meredith, Ahern,
  Pugmire, Biagas, Miller, Harrison, Weber, Krishnan, Fogal, Sanderson, Garth,
  Bethel, Camp, R\"{u}bel, Durant, Favre, and Navr\'{a}til}]{HPV_VisIt}
\bibinfo{author}{H.~Childs}, \bibinfo{author}{E.~Brugger},
  \bibinfo{author}{B.~Whitlock}, \bibinfo{author}{J.~Meredith},
  \bibinfo{author}{S.~Ahern}, \bibinfo{author}{D.~Pugmire},
  \bibinfo{author}{K.~Biagas}, \bibinfo{author}{M.~Miller},
  \bibinfo{author}{C.~Harrison}, \bibinfo{author}{G.~H. Weber},
  \bibinfo{author}{H.~Krishnan}, \bibinfo{author}{T.~Fogal},
  \bibinfo{author}{A.~Sanderson}, \bibinfo{author}{C.~Garth},
  \bibinfo{author}{E.~W. Bethel}, \bibinfo{author}{D.~Camp},
  \bibinfo{author}{O.~R\"{u}bel}, \bibinfo{author}{M.~Durant},
  \bibinfo{author}{J.~M. Favre}, \bibinfo{author}{P.~Navr\'{a}til},
\newblock \bibinfo{title}{{VisIt: An End-User Tool For Visualizing and
  Analyzing Very Large Data}},
\newblock in: \bibinfo{booktitle}{{High Performance Visualization--Enabling
  Extreme-Scale Scientific Insight}}, \bibinfo{year}{2012}, pp.
  \bibinfo{pages}{357--372}.

\end{thebibliography}

\end{document}